\documentclass[10pt,twoside,a4paper]{article}         

\usepackage[T1]{fontenc}
\usepackage[british,UKenglish,USenglish,english,american]{babel}
\usepackage[pdftex]{graphicx}           
\usepackage{setspace}                   
\usepackage{indentfirst}               
\usepackage{courier}                    
\usepackage{type1cm}                    
\usepackage{listings}                  
\usepackage{titletoc}
\usepackage[fixlanguage]{babelbib}
\usepackage[font=small,format=plain,labelfont=bf,up,textfont=it,up]{caption}
\usepackage[usenames,svgnames,dvipsnames]{xcolor}
\usepackage[a4paper,top=2.54cm,bottom=2.0cm,left=2.0cm,right=2.54cm]{geometry}
\usepackage[pdftex,plainpages=false,pdfpagelabels,colorlinks=true,citecolor=DarkGreen,linkcolor=NavyBlue,urlcolor=DarkRed,filecolor=green,bookmarksopen=true]{hyperref} 
\usepackage[all]{hypcap}                
\fontsize{60}{62}\usefont{OT1}{cmr}{m}{n}{\selectfont}
\usepackage{dsfont}
\usepackage{amsthm}
\usepackage{amsfonts}
\usepackage{amsmath}
\usepackage{amssymb}
\usepackage{mathrsfs}
\usepackage{lmodern}
\usepackage{relsize}
\usepackage{url}
\theoremstyle{plain}
\usepackage{mathtools}
\usepackage{thmtools}
\usepackage{enumitem}
\usepackage{calrsfs}
\usepackage{epsfig}
\usepackage{amsthm}

\newtheorem{teo}{Theorem}[section]

\newtheorem{prop}{Proposition}[section]
\newtheorem{lemma}{Lemma}[section]
\newtheorem{cor}{Corollary}[section]

\newtheorem{condition}{Condition}[section]
\newtheorem{definition}{Definition}[section]
\theoremstyle{definition}

\numberwithin{equation}{section}
\newcommand{\z}{\mathbb{Z}}

\title{Nucleation for one-dimensional long-range Ising models}
\author{Aernout C.D. van Enter\\
\footnotesize{ Bernoulli Institute for Mathematics, Computer Science and Artificial Intelligence}\\
\footnotesize{\texttt{avanenter@gmail.com}}\\
\footnotesize{ Groningen University, Nijenborgh 9, 9747AG, Groningen, The Netherlands}\\
\footnotesize{and Delft Institute of Applied Mathematics}\\
\footnotesize{Technische Universiteit Delft, Van Mourik Broekmanweg 6, 2628 XE Delft, The Netherlands}\\
[0.3cm] Bruno Kimura\\
\footnotesize{Delft Institute of Applied Mathematics}\\
\footnotesize{\texttt{bkimura@tudelft.nl}}\\
\footnotesize{Technische Universiteit Delft, Van Mourik Broekmanweg 6, 2628 XE Delft, The Netherlands}\\
[0.3cm] Wioletta Ruszel\\
\footnotesize{Delft Institute of Applied Mathematics}\\
\footnotesize{\texttt{W.M.Ruszel@tudelft.nl}}\\
\footnotesize{Technische Universiteit Delft, Van Mourik Broekmanweg 6, 2628 XE Delft, The Netherlands}\\
[0.3cm] Cristian Spitoni\\
\footnotesize{Institute of Mathematics}\\
\footnotesize{\texttt{C.Spitoni@uu.nl}}\\
\footnotesize{Universiteit Utrecht, Budapestlaan 6, 3584 CD Utrecht, The Netherlands}\\}
\date{\today}

\begin{document}

\maketitle

\begin{abstract}
In this note we study metastability phenomena for a class of long-range Ising models in one-dimension. We prove that, under suitable general conditions, the configuration $\mathbf{-1}$ is the only metastable state and we estimate the mean exit time. 
Moreover,  we illustrate the theory with two examples (exponentially and polynomially decaying interaction) and we show that the critical droplet can  be  \emph{macroscopic} or \emph{mesoscopic}, according to the value of the external magnetic field. 
\end{abstract}

\section{Introduction}
Metastability is a dynamical phenomenon observed in many different contexts, such as physics, chemistry, biology, climatology, economics.
Despite the variety of scientific areas,  the common feature of all these situations is the existence of multiple,
well-separated \emph{time scales}. On short time scales the system is in a
quasi-equilibrium within a single region, while on long time scales it undergoes
rapid transitions between quasi-equilibria in different regions.   A rigorous description of metastability in the setting of  stochastic dynamics is relatively recent, dating back to
the pioneering paper \cite{CGOV}, and has experienced substantial progress in the last decades. See \cite{BL,Bo, BdH,OV} for reviews and for a list of the most important papers on this subject.

One of the big challenges in rigorous study of metastability is understanding the dependence of  the metastable behaviour  and of the nucleation process of the stable phase 
on the dynamics.  The nucleation process of the critical droplet, i.e. the configuration triggering the crossover, has been indeed studied in different dynamical regimes: serial (\cite{BM, CO}) vs. 
parallel dynamics (\cite{BCLS,CN,CNS01}); non-conservative (\cite{BM, CO}) vs. conservative dynamics  (\cite{HNT,HNT1,HOS}); finite (\cite{BHN}) vs. infinite volumes (\cite{BHS}); competition (\cite{CNS02,CNS03,I,SLF}) vs. non-competition of metastable phases (\cite{CN2013,CNS2015}).  All previous studies assumed that the microscopic interaction is of short-range type.

 The present paper pushes further this investigation, studying the dependence of the metastability scenario on the  \emph{range} of the interaction of the model. Long-range Ising models in low dimensions are known to behave like higher-dimensional short-range models. For instance in \cite{Dys, Cas}  (and later generalized by \cite{Picco, WBruno}) it was shown that long-range Ising models undergo a phase transition already in one dimension, and this transition persists in fast enough  decaying fields. Furthermore, Dobrushin interfaces are rigid already in two dimensions for anisotropic long-range Ising models, see \cite{Loren}.

We consider the question: does indeed a  \emph{long-range} interaction change substantially the nucleation process? Are we able to define in this framework a critical configuration triggering the crossover towards the stable phase?
In (\cite{MCAW}) the author already considered the \emph{Dyson-like}  long-range models, i.e. the one-dimensional lattice model of Ising spins with interaction decaying with a power $\alpha$, in a external magnetic field. Despite the long-range potential, the author showed, by \emph{instanton} arguments, that the system has a finite-sized critical droplet.

In this manuscript we want to make rigorous this claim for a general long-range interaction, showing  as well that the long-range interaction completely changes the metastability scenario: in the short--range one-dimensional Ising model a droplet of size one,  already nucleates the stable phase. 
We show instead that for a given external field $h$, and  pair long range potential $J(n)$, we can define a nucleation droplet which gets larger for smaller $h$. For $d=1$ finite range  interactions, inserting a minus interval of size $\ell$ in the plus phase costs a finite energy, which is uniform in the length of the interval, the same is almost true for a fast decaying interaction, as there is a uniform bound on the energy an interval costs.  Thus, for low temperature, there is a diverging timescale and we will talk  in case (maybe by abuse of terminology) of metastability.
The spatial scale of a nucleating interval, however,  defined as an interval which lowers its energy when growing, is finite for finite range interactions, but diverges as  $h\to0$ for infinite range.
The Dyson model has energy and spatial scale of nucleating droplet diverging as $h$ goes to zero. 
We will show that, depending on the value of $h$, the critical droplet can be \emph{macroscopic} or \emph{mesoscopic}.
Roughly speaking, an interval of minuses of length $\ell$ which grows to $\ell+1$ gains energy $2h$, but loses $E_\ell= \sum_{ n=\ell }^{\infty} J(n)$. $E_\ell$ converges to zero as $\ell\rightarrow \infty$, but the smaller $h$ is, the larger the size of the critical droplet. Moreover, by taking $h$ volume-dependent, going to zero with $N$ as $ N^{-\delta}$, one can make the nucleation interval mesoscopic (e.g. $O(N^\delta)$, with 
$\delta\in(0,1)$) or macroscopic (i.e. $O(N)$).

The paper is organised as follows. In Section~2 we describe the lattice model and we give the main definitions; in Section~3 the main results of the paper are stated, while in Section~4 and 5 the proofs of the model-dependent results are given.

\section{The model and main definitions}

 Let $\Lambda$ be a finite interval of $\z$, and let us denote by $h$ a positive external field. Given a configuration 
$\sigma$ in $\Omega_{\Lambda} = \{-1,1\}^{\Lambda}$, we define the  \textit{Hamiltonian} with respect to free boundary condition by
\begin{equation}\label{ham1}
H_{\Lambda, h}(\sigma) = -\sum_{\{i,j\} \subseteq \Lambda} J(|i-j|)\sigma_{i}\sigma_{j} - \sum_{i \in \Lambda}h \sigma_{i}, 	
\end{equation}
where $J: \mathbb{N} \rightarrow \mathbb{R}$, the \textit{pair interaction}, is assumed to be positive and decreasing. The class of interactions that we want to include in the present analysis are of \textit{long-range type}, for instance,
\begin{enumerate}
\item exponential decay: $J(|i-j|)= J \cdot \lambda^{-|i-j|}$  with constants $J>0$ and $\lambda >1$;
\item polynomial decay: $J(|i-j|)= J \cdot |i-j|^{-\alpha}$, where $\alpha>0$ is a parameter.
\end{enumerate} 
The {\it finite-volume Gibbs measure} will be denoted by
\begin{equation}
\label{e:gibbs}
\mu_{\Lambda}(\sigma) = \frac{1}{Z_{\Lambda}} \exp \left(-\beta H_{\Lambda, h}(\sigma)\right ),
\end{equation}
where $\beta>0$ is proportional to the inverse temperature and $Z_{\Lambda}$  is a normalizing constant. The set of \emph{ground states} $\mathscr{X}^{s}$ is defined as 
$\mathscr{X}^{s}: =\textnormal{argmin}_{\sigma\in \Omega_\Lambda} H_{\Lambda, h}(\sigma)$. Note that for the class of interactions considered 
 $\mathscr{X}^{s} = \{\mathbf{+1}\}$,  where $\mathbf{+1}$ {stands for} the configuration with all spins equal to $+1$.
\\
Given an integer $k \in \{0, \dots, \#\Lambda\}$, we consider  $\mathcal{M}_k:=\{\sigma\in \Omega_{\Lambda}:  \# \{i: \sigma_{i} = 1\} = k\}$ consisting of configurations in $\Omega_{\Lambda}$ with $k$
positive spins, and we define the configurations $L^{(k)}$ and $R^{(k)}$ as follows. Let
\begin{equation}
L_{i}^{(k)} =
\begin{cases}
+1 & \text{if $1 \leq i \leq  k $, and}\\
-1 & \text{otherwise,}	
\end{cases}	
\end{equation}
and
\begin{equation}
R_{i}^{(k)} =
\begin{cases}
-1 & \text{if $1 \leq i \leq \# \Lambda-k$, and}\\
+1 & \text{otherwise,}	
\end{cases}	
\end{equation}
i.e.,  the configurations respectively with $k$ positive spins on {\it left} side of the interval and on the {\it right} one. We will show that $L^{(k)}$ and $R^{(k)}$ are the {minimizers} of the energy 
function $H_{\Lambda,h}$ on $\mathcal{M}_k$ (see {Proposition}~\ref{csgo2}). Let us denote {by $\mathcal{P}^{(k)}$ the set $\mathcal{P}^{(k)}:= \{L^{(k)}, R^{(k)}\}$ consisting} of the {minimizers} of the energy on $\mathcal{M}_k$. With abuse of notation we will indicate with 
$H_{\Lambda,h}(\mathcal{P}^{(k)})$ the energy of the elements of the set, {that is, $H_{\Lambda,h}(\mathcal{P}^{(k)}):=H_{\Lambda,h}({L}^{(k)})=H_{\Lambda, h}({R}^{(k)})$.}

We {choose} the evolution of the system {to be described} by a discrete-time Markov chain {$X=(X(t))_{t \geq 0}$}, in particular, we consider the discrete-time serial Glauber dynamics given by the Metropolis weights, i.e.,
{the transition matrix of such dynamics is given by}
$$
p(\sigma,\eta):= c(\sigma,\eta)e^{-\beta[H_{\Lambda,h}(\eta)-H_{\Lambda,h}(\sigma)]_+},
$$
where  $[\cdot]_+$ denotes the positive part, and $c(\cdot,\cdot)$ is its connectivity matrix {that is equal to
$1/|\Lambda|$ in case the two configurations $\sigma$ and $\eta$ coincide up to the value of a single spin, and zero otherwise.} Notice that such dynamics is reversible with respect to the Gibbs measure defined in 
(\ref{e:gibbs}). Let us define the \emph{hitting time} {$\tau_{\eta}^{\sigma}$} of a configuration $\eta$ of the chain  $X$ started at $\sigma$ as
\begin{equation}
\label{e:hit}
{\tau_{\eta}^{\sigma}:=\inf\{t>0: X(t)=\eta\}.}
\end{equation}

For any positive integer $n$, {a sequence $\gamma = (\sigma^{(1)}, \dots, \sigma^{(n)})$
such that $\sigma^{(i)}\in \Omega_\Lambda $  and $c(\sigma^{(i)},\sigma^{(i+1)})>0$} for all $i=1,\dots,n-1$
is called a \emph{path} joining {$\sigma^{(1)}$ to $\sigma^{(n)}$};
we also say that $n$ is the length of the path.
For any path {$\gamma$} of length $n$, we let
\begin{equation}
\label{height}
{\Phi_\gamma :=\max_{i=1,\dots,n} H_{\Lambda, h}(\sigma^{(i)})}
\end{equation}
be the \emph{height} of the path. {We also define
the \emph{communication height}
between $\sigma$ and $\eta$ by
\begin{equation}
\label{communication}
\Phi(\sigma,\eta)
:=
\min_{\gamma\in\Omega(\sigma,\eta)}
\Phi_\gamma,
\end{equation}
where the minimum is restricted to the set $\Omega(\sigma,\eta)$ of all paths joining $\sigma$ to $\eta$.}
By reversibility, it easily follows that
\begin{equation}
\label{rev02}
\Phi(\sigma,\eta)=\Phi(\eta,\sigma)
\end{equation}
for all $\sigma,\eta\in \Omega_\Lambda$.
{We extend the previous definition for sets $\mathcal{A},\mathcal{B}\subseteq \Omega_\Lambda$ by letting}
\begin{equation}
\label{communication-set}
\Phi(\mathcal{A},\mathcal{B})
:=
{\min_{\gamma\in\Omega(\mathcal{A},\mathcal{B})}\Phi_\gamma}
=
\min_{\sigma\in \mathcal{A},\eta\in \mathcal{B}}\Phi(\sigma,\eta),
\end{equation}
where {$\Omega(\mathcal{A},\mathcal{B})$ denotes} the set of paths joining
a state in $\mathcal{A}$ to a state in $\mathcal{B}$.
The {\it communication cost} of passing from 
$\sigma$ to $\eta$ is given by the quantity $\Phi(\sigma,\eta)-H_{\Lambda, h}(\sigma)$.
{Moreover, if we define $\mathscr{I}_\sigma$ as the set of all states $\eta$ in $\Omega_\Lambda$
such that $H_{\Lambda, h}(\eta)< H_{\Lambda, h}(\sigma)$, then the
\emph{stability level} of any $\sigma\in \Omega_\Lambda \setminus \mathscr{X}^{s}$ is given by
\begin{equation}
\label{stability}
V_\sigma:=\Phi(\sigma,\mathscr{I}_\sigma)-H_{\Lambda, h}(\sigma)
\ge 0.
\end{equation}
}
Following \cite{MNOS}, we now introduce the notion of
\emph{maximal stability level}.
{Assuming that} $\Omega_\Lambda\setminus \mathscr{X}^{s}\neq\emptyset$, we let
the \emph{maximal stability level} be
\begin{equation}
\label{gamma}
\Gamma_\textnormal{m}:=\sup_{\sigma\in \Omega_\Lambda\setminus \mathscr{X}^{s}}V_\sigma.
\end{equation}
We give the following definition.
\begin{definition}
\label{def1}
 We call metastable set $\mathscr{X}^{m}$, the set
\begin{equation}
\label{metastabile}
\mathscr{X}^{m}
:=
\{\sigma\in \Omega_\Lambda\setminus \mathscr{X}^{s}:\,V_\sigma=\Gamma_{\textnormal{m}}\}.
\end{equation}
\end{definition}

Following \cite{MNOS}, we shall call
$\mathscr{X}^{m}$ the set of \emph{metastable} states of the
system {and refer to each of its elements as \emph{metastable}.}
We denote {by} $\Gamma$ the quantity
\begin{equation}
\label{e:gamma}
{\Gamma:=\max_{k=0, \dots, \#\Lambda} H_{\Lambda,h}(\mathcal{P}^{(k)})- H_{\Lambda,h}(\mathbf{-1}).}
\end{equation}
We will show in Corollary~\ref{t:meta} that {under certain assumptions} $\Gamma=\Gamma_\textnormal{m}$.

\section{Main Results}
\subsection{Mean exit time}
In this section we will study the first hitting time of the configuration $\mathbf{+1}$ when the system is prepared in $\mathbf{-1}$, in the limit $\beta\to \infty$. 
We will restrict our analysis to the case given by the following condition.
\begin{condition}
\label{c:condition}
{Let $N$ be an integer such that $N \geq 2$. We consider $\Lambda = \{1, \dots, N\}$ and $h$ such that
\begin{equation}\label{field}
0 < h < \sum_{n = 1}^{N-1}J(n).
\end{equation} }
\end{condition}
By using the general theory developed in \cite{MNOS}, we need first to solve two \emph{model-dependent} problems: the calculation of the \emph{minimax} 
between {$\mathbf{-1}$} and $\mathbf{+1}$ ({item} \ref{minmax1} of Theorem~\ref{t:minmax})
and the proof of a 
 \emph{recurrence} property in the energy landscape ({item} \ref{minmax2}  of Theorem~\ref{t:minmax}).
\begin{teo}
\label{t:minmax}
Assume that Condition~\ref{c:condition} is satisfied.{Then, we have}
\begin{enumerate}
\item $\Phi(-\mathbf{1},\mathbf{+1})=\Gamma+H_{\Lambda,h}(\mathbf{-1})$, \label{minmax1}
\item {$V_{\mathbf{-1}}= \Gamma > 0$}, and \label{minmax3}
\item $V_\sigma<\Gamma$ for any {$\sigma\in \Omega_\Lambda\setminus \{\mathbf{-1}, \mathbf{+1}\}$.\label{minmax2}}
\end{enumerate}
\end{teo}
As a corollary we have that $-\mathbf{1}$ is the only metastable state for this model.
\begin{cor}
\label{t:meta}
Assume that Condition~\ref{c:condition} is satisfied. {It follows that 
\begin{equation}
\Gamma = \Gamma_m ,	
\end{equation}
and
\begin{equation}
\mathscr{X}^{m}=\{-\mathbf{1}\}.	
\end{equation}
}
\end{cor}

Therefore, the asymptotic of the exit time for the system
started at the metastable states {is given by the following theorem.}
\begin{teo}
\label{t:meantime}
Assume that Condition~\ref{c:condition} is satisfied. {It follows that}
\begin{enumerate}
\item for any $\epsilon>0$
$$
{\lim_{\beta\to\infty} \mathbb{P}\left(e^{\beta(\Gamma-\epsilon)}<\tau_{\mathbf{+1}}^{\mathbf{-1}}<e^{\beta(\Gamma+\epsilon)}\right)=1,}
$$
\item the limit
$$
{\lim_{\beta\to\infty}\frac{1}{\beta}\log\left(\mathbb{E}\left(\tau_{\mathbf{+1}}^{\mathbf{-1}}\right)\right)=\Gamma}
$$
holds.	
\end{enumerate}
\end{teo}

Once the model-dependent results in Theorem~\ref{t:minmax} have been proven, the proof of Theorem~\ref{t:meantime} easily follows from the general 
theory present in \cite{MNOS}: {item} 1  follows from Theorem~4.1 in \cite{MNOS} and {item} 2 from Theorem~4.9 in \cite{MNOS}.

\subsection{Nucleation of the metastable phase}
We are going to show that for small enough external magnetic field, the size of the critical droplet is a macroscopic fraction of the system, while for  $h$ sufficiently large,  the critical configuration will be a {mesoscopic} fraction of the system.

Let us define
$L := \left\lfloor \frac{N}{2} \right\rfloor$, and {let $h_{k}^{(N)}$ be}
\begin{equation}
{h_{k}^{(N)} := \sum_{n=1}^{N-k-1} J(n) - \sum_{n=1}^{k} J(n)}   
\end{equation}
for each $k = 0,\dots, L-1$. {One can easily verify that 
\begin{equation}
0 < h_{L-1}^{(N)} < \dots < h_{1}^{(N)} < h_{0}^{(N)} = \sum_{n = 1}^{N-1}J(n)
\end{equation}
}
\begin{prop}\label{critdrop}
{Under the assumption that Condition (\ref{c:condition}) is satisfied, one of the following conditions holds.
\begin{enumerate}
 \item Case $h < h_{L-1}^{(N)}$, we have 
 \[H_{\Lambda,h}(\mathcal{P}^{(L)}) > \max_{\substack{0 \leq k \leq N \\ k \neq L}} H_{\Lambda,h}(\mathcal{P}^{(k)}).\]
 
 \item Case $h_{k}^{(N)} < h < h_{k-1}^{(N)}$ for some $k \in \{1,\dots, L-1\}$, we have  
 \[H_{\Lambda,h}(\mathcal{P}^{(k)}) > \max_{\substack{0 \leq i \leq N \\ i \neq \bar{k}}} H_{\Lambda,h}(\mathcal{P}^{(i)}).\]

 \item Case $h = h_{k}^{(N)}$ for some $k \in \{1,\dots, L-1\}$, we have 
 \[H_{\Lambda,h}(\mathcal{P}^{(k)}) = H_{\Lambda,h}(\mathcal{P}^{(k+1)}) > \max_{\substack{0 \leq i \leq N \\ i \neq k, i \neq k+1}} H_{\Lambda,h}(\mathcal{P}^{(i)}).\]
\end{enumerate}
}
\end{prop}

The first point of Proposition~\ref{critdrop} describes the less interesting and, in a way,  artificial, situation of very low external magnetic fields: 
in this regime the \emph{bulk} term is negligible so that the energy of the droplet increases until the positive spins are the majority
(i.e. {$k=L$}, see Figure~\ref{fig:test2}). Therefore, {the second point} contains the most interesting situation, where there is  an interplay between the 
bulk and the \emph{surface} term. 
The following Corollary is a consequence of Proposition~\ref{critdrop}  when $N$ is large enough and gives a characterisation of the critical size $k_c$ of the critical droplet.

\begin{cor}\label{corcrit}
{If we assume that $\sum_{n = 1}^{\infty}J(n)$ converges and 
\begin{equation}
0 < h < \sum_{n = 1}^{\infty}J(n),	
\end{equation} 
then, the size of the critical droplet will be given by
\begin{equation}\label{crit}
k_{c} = \min\left\{k \in \mathbb{N}: \sum_{n = k+1}^{\infty}J(n) \leq h\right\}
\end{equation}
whenever $N$ is sufficiently large.}
\end{cor}
 
As a consequence {of Corollary \ref{corcrit}}, the set of \emph{critical configurations} $\mathcal{P}_c$  is given by {
\begin{equation}
\label{critical_conf}
\mathcal{P}_c:= \{L^{(k_c)}, R^{(k_c)}\}
\end{equation}
for $N$ large enough.} {The following result shows the reason why configurations in $\mathcal{P}_c$ are referred to as
\emph{critical} configurations: they indeed trigger the transition towards the stable phase.}
\begin{lemma}
\label{nucleation_1}
Under the conditions stated above,  we have
\begin{enumerate}
\item any path {$\gamma\in \Omega(\mathbf{-1},\mathbf{+1})$ such that $\Phi_\gamma- H_{\Lambda,h}(-\mathbf{1})=\Gamma$} visits $\mathcal{P}_c$, and
\item {the limit}
{
$$
\lim_{\beta\to\infty} \mathbb{P}(\tau_{\mathcal{P}_c}^{-\mathbf{1}}<\tau_{+\mathbf{1}}^{-\mathbf{1}})=1
$$
holds.}
\end{enumerate}
\end{lemma}
The proof of the previous Theorem is a straightforward consequence of Theorem~5.4 in \cite{MNOS}.

\subsection{Examples}
Let us give two interesting examples of the general theory so far developed.
\subsubsection{Example 1: exponentially decaying coupling}
We consider 
{$$J(n) = \frac{J}{\lambda^{n-1}},$$}
where $J$ and $\lambda$ are positive real numbers with $\lambda > 1$. 
\begin{prop}
\label{expo}
{Under the same hypotheses as Corollary~\ref{corcrit}, we have that the critical droplet length $k_c$ is equal to 
\begin{equation}
 k_{c} = \left\lceil \log_{\lambda}\left(\frac{J}{h(1-\lambda^{-1})}\right) \right\rceil
\end{equation}
whenever N is sufficiently large.}
\end{prop}
\begin{proof}
{By Corollary~\ref{corcrit}, we have
\[ J \sum_{n = k_{c}+1}^{\infty}\lambda^{-(n-1)} \leq h < J \sum_{n = k_{c}}^{\infty}\lambda^{-(n-1)}\]
that implies
\[\frac{\lambda^{-k_{c}}}{1-\lambda^{-1}} \leq \frac{h}{J} < \frac{\lambda^{-(k_{c}-1)}}{1-\lambda^{-1}}\]
Thus
\begin{equation}
k_{c} - 1 < - \frac{\log\left(\frac{h(1-\lambda^{-1})}{J}\right)}{\log \lambda} \leq k_{c}.
\end{equation}
}
\end{proof}
{As a remark we notice that in case of exponential decay of the interaction, the system behaves essentially  as the nearest neighbours one-dimensional
 Ising model. Note that   
\begin{equation}
\lim_{\lambda \to \infty}  J(n) = 
\begin{cases}
J &\text{if $n =1$, and}\\
0 &\text{otherwise;}	
\end{cases}
\end{equation}
moreover, if $h < J = \lim_{\lambda \to \infty}  \sum_{n =1}^{\infty}J(n)$, then $k_{c} = 1$ whenever $\lambda$ is large enough. So, 
we conclude that typically a single plus spin in the lattice will trigger the nucleation of the stable phase.
 As you can see in Figure~\ref{expo22} the energy exitations 
$H_{\Lambda,h}(\mathcal{P}^{(k)})-H_{\Lambda,h}(\mathbf{-1})$ are strictly descreasing in $k$, as expected.}
\begin{figure}
  \centering
  \includegraphics[height=7cm]{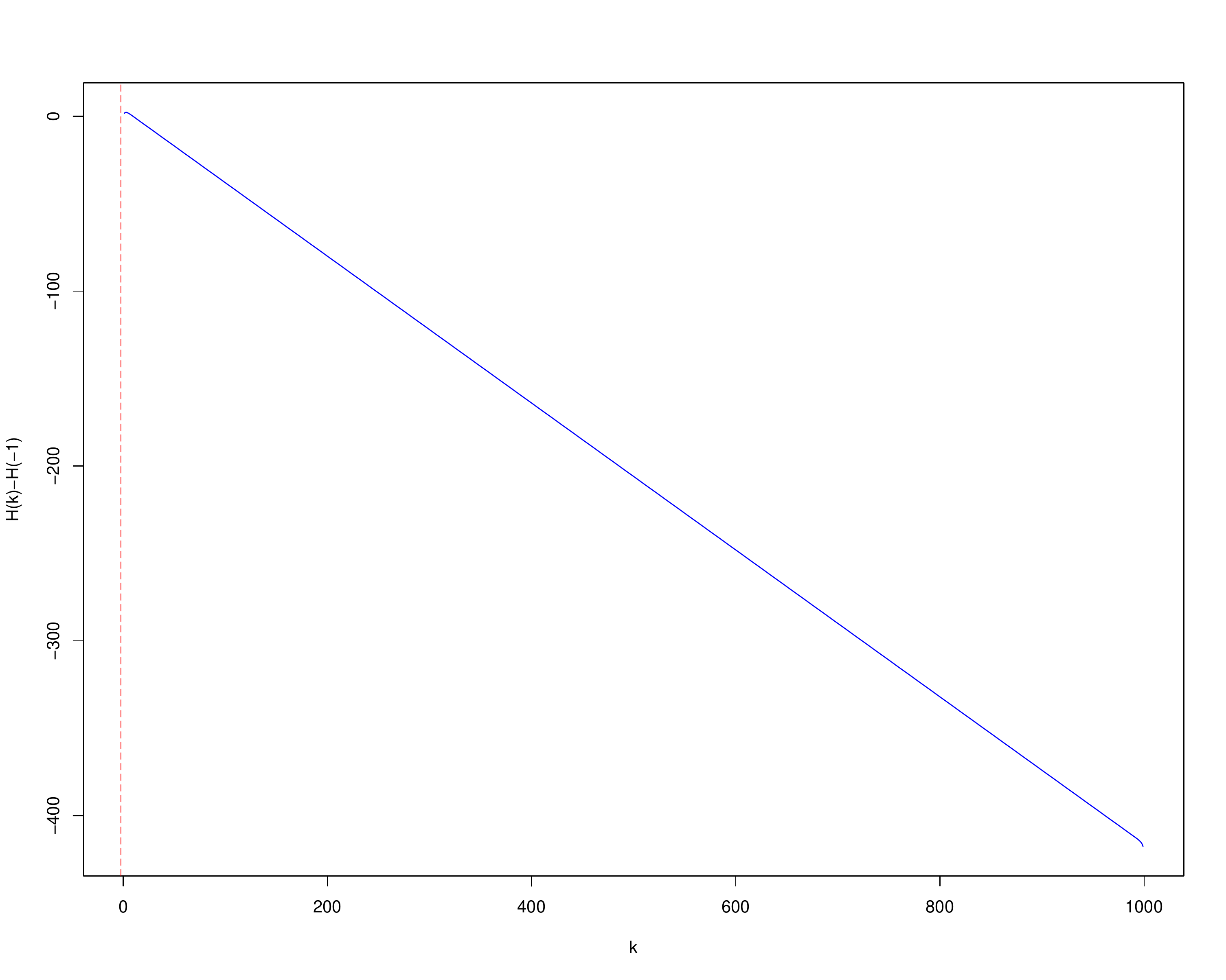}
  \caption{Blue line is the excitation energy 
$H_{\Lambda,h}(\mathcal{P}^{(k)})-H_{\Lambda,h}(\mathbf{-1})$ for $N=1000$,
$\lambda=2, h=0.21,J=1$; 
 red line is the critical droplet.}
\label{expo22}
\end{figure}

\subsubsection{Example 2: polynomially decaying coupling}
 Let the coupling constants be given by  
 $$J(n) = J\cdot n^{-\alpha},$$
 where $J$ and $\alpha$ are positive real numbers with $\alpha > 1$.  As it is shown in Figures~\ref{fig:test1} and  \ref{fig:test2}, for the polynomially decaying coupling model, 
 we have that, for $h$ small enough the critical droplet is essentially the half interval, while for large enough magnetic  external magnetic field, the critical droplet is the configuration with $k_c$ plus spins at the sides, with 
 $k_c\approx  \left(\frac{J}{h(\alpha -1)}\right)^{\frac{1}{\alpha -1}} $. 
 \begin{figure}
\centering
\begin{minipage}{.5\textwidth}
  \centering
  \includegraphics[width=.9\linewidth]{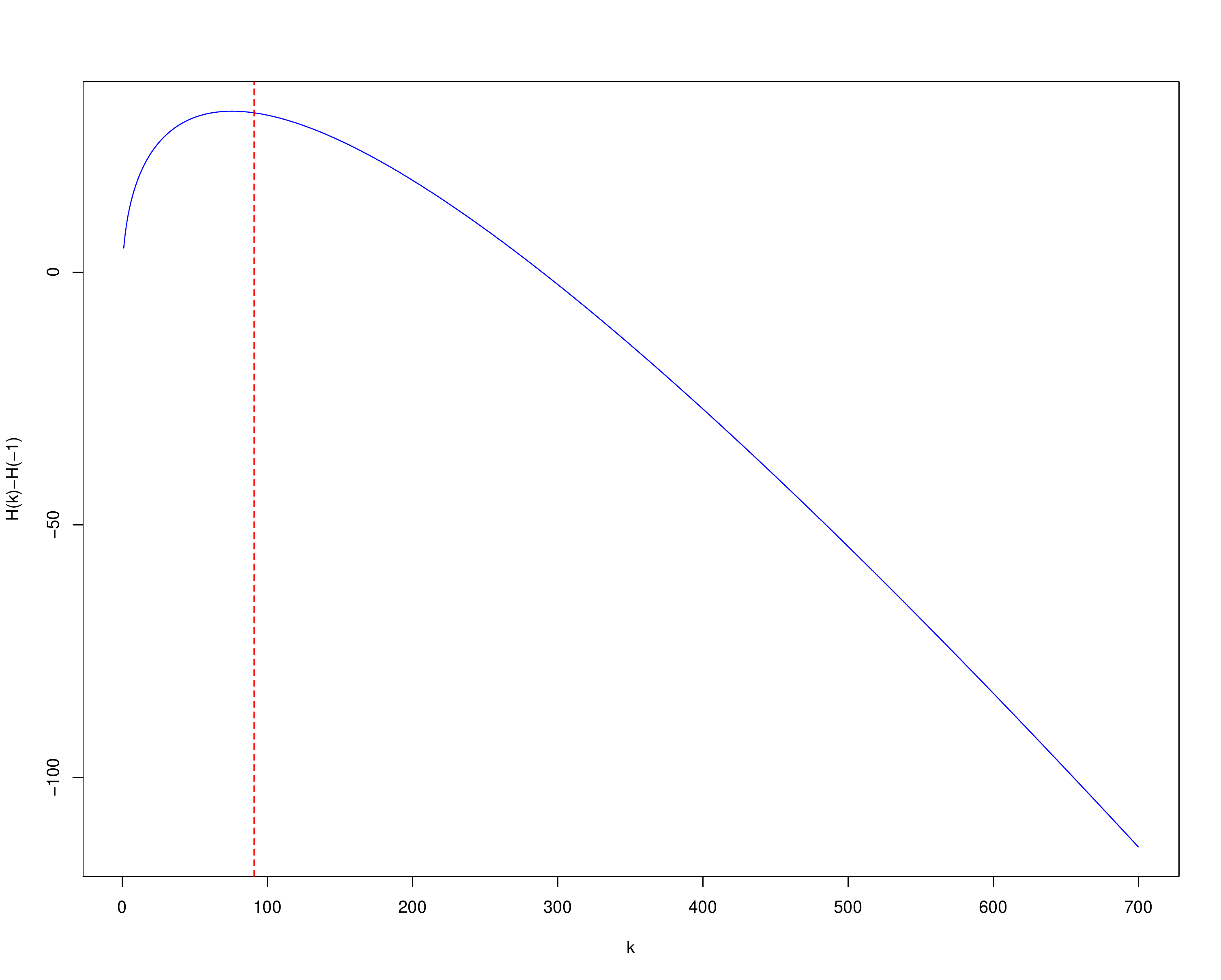}
  \captionof{figure}{Blue line is the excitation energy \\
$H_{\Lambda,h}(\mathcal{P}^{(k)})-H_{\Lambda,h}(\mathbf{-1})$ for $N=10000$,\\
$\alpha=3/2, h=0.21,J=1$; {the
 red line represents the \\ critical length $k_c\approx 91$.}}
  \label{fig:test1}
\end{minipage}%
\begin{minipage}{.5\textwidth}
  \centering
  \includegraphics[width=.9\linewidth]{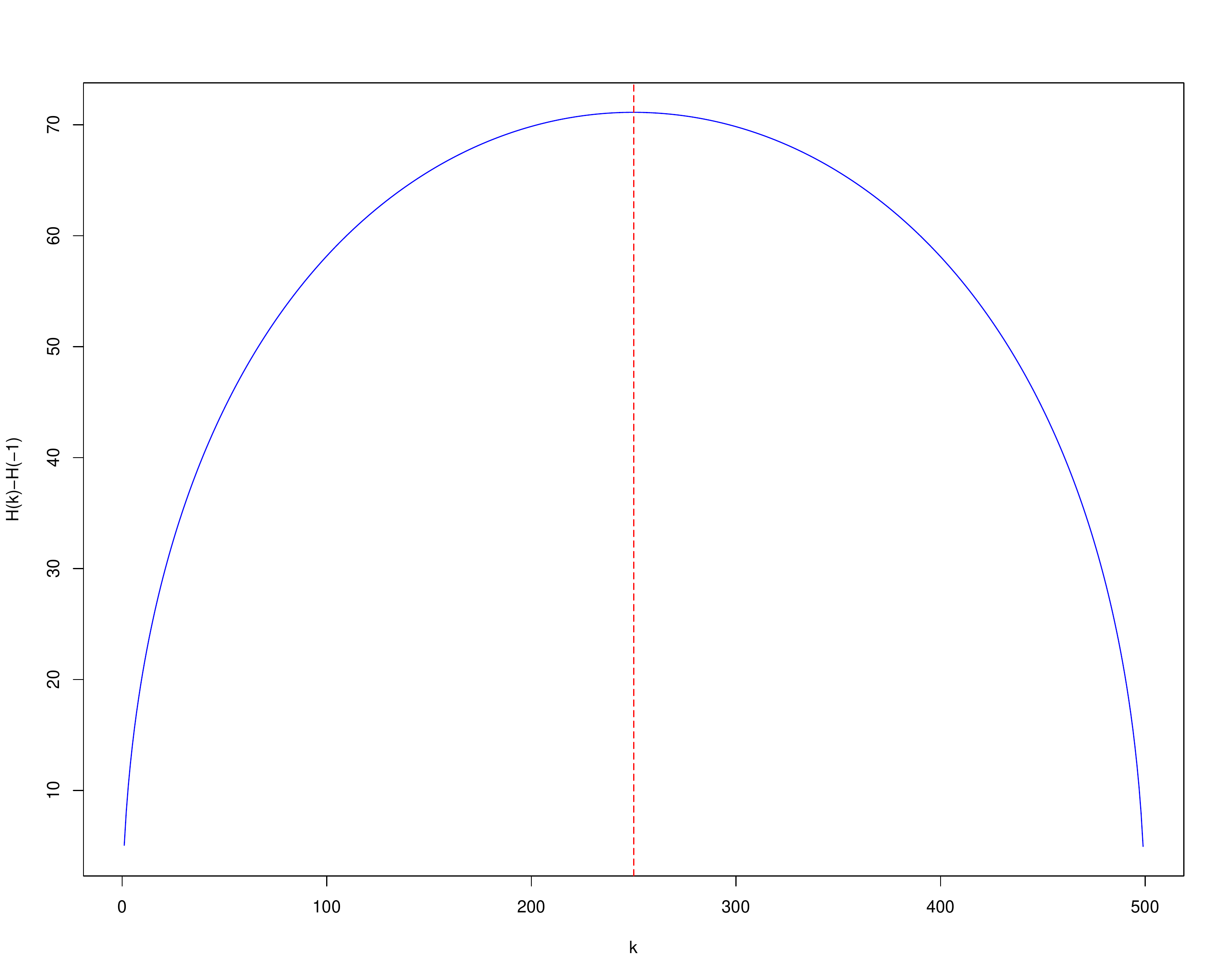}
  \captionof{figure}{
Blue line is the excitation energy \\
$H_{\Lambda,h}(\mathcal{P}^{(k)})-H_{\Lambda,h}(\mathbf{-1})$ for $N=500$,\\
$\alpha=3/2, h=0.0001,J=1$; {the
 red line represents the \\critical length $k_c= 250$.}}
  \label{fig:test2}
\end{minipage}
\end{figure}
 
 We can prove indeed the following proposition.
\begin{prop}\label{dyson}
{Under the same hypotheses as Corollary~\ref{corcrit}, we have that $k_c$ satisfies
\begin{equation}
\left | k_{c} - \left(\frac{J}{h(\alpha -1)}\right)^{\frac{1}{\alpha -1}} \right |<    1
\end{equation}
whenever $N$ is large enough.}
\end{prop}
\begin{proof}
{By Corollary~\ref{corcrit}}, it follows that
$$
 J \sum_{n = k_{c}+1}^{\infty}n^{-\alpha} \leq h < J \sum_{n = k_{c}}^{\infty}n^{-\alpha}.
$$
Moreover, note that
$$
 \int_{k_{c}+1}^{\infty}\frac{1}{x^{\alpha}} dx < \sum_{n = k_{c}+1}^{\infty}n^{-\alpha}
$$
and
\[ \sum_{n = k_{c}}^{\infty}n^{-\alpha} <  \int_{k_{c}-1}^{\infty}\frac{1}{x^{\alpha}} dx \]
so that
\[\frac{(k_{c}+1)^{1-\alpha}}{\alpha - 1} < \frac{h}{J} < \frac{(k_{c}-1)^{1-\alpha}}{\alpha - 1}. \]
Hence,
\begin{equation}
 (k_{c} - 1)^{\alpha - 1} < \frac{J}{h(\alpha -1)} < (k_{c} + 1)^{\alpha - 1} .
\end{equation}
\end{proof}


\section{Proof Theorem~\ref{t:minmax}}
We start the proof of the main theorem giving some general results about the control of the energy of a general configuration.
First of all we note that  equation (\ref{ham1}) can be written as
\begin{eqnarray*}
H_{\Lambda, h}(\sigma) &=& -\frac{1}{2}\sum_{i \in \Lambda}\sum_{j \in \Lambda} J(|i-j|)\sigma_{i}\sigma_{j} - h\sum_{i \in \Lambda} \sigma_{i}\\	
&=& \sum_{i \in \Lambda}\sum_{j \in \Lambda} J(|i-j|)\left(\frac{1 - \sigma_{i}\sigma_{j}}{2}\right) - h\sum_{i \in \Lambda} \sigma_{i} - \frac{1}{2}\sum_{i \in \Lambda}\sum_{j \in \Lambda} J(|i-j|)\\
&=& \sum_{i \in \Lambda}\sum_{j \in \Lambda} J(|i-j|) \mathds{1}_{\{\sigma_{i} \neq \sigma_{j}\}} - h\sum_{i \in \Lambda}\sigma_{i} - \frac{1}{2}\sum_{i \in \Lambda}\sum_{j \in \Lambda} J(|i-j|).
\end{eqnarray*}
Moreover, given an integer $k \in \{0, \dots, N\}$, if $\sigma\in\mathcal{M}_k$, then
\begin{equation}\label{main}
H_{\Lambda, h}(\sigma) = \sum_{i \in \Lambda}\sum_{j \in \Lambda} J(|i-j|) \mathds{1}_{\{\sigma_{i} \neq \sigma_{j}\}}   + h (N - 2k) - \frac{1}{2}\sum_{i \in \Lambda}\sum_{j \in \Lambda} J(|i-j|).
\end{equation} 
Therefore, restricting ourselves to configurations that contains only $k$ spins with the value $1$, in order to find such configurations with minimal energy, 
it is sufficient to
minimize the first term of the right-hand side of equation (\ref{main}).

\begin{prop}\label{csgo}
Let $N$ be a positive integer and $k \in \{0, \dots, N\}$, if we restrict to all $\sigma\in\mathcal{M}_k$, then	
\begin{equation}\label{csgo1}
\sum_{i=1}^{N}\sum_{j=1}^{N} J(|i-j|) \mathds{1}_{\{\sigma_{i} \neq \sigma_{j}\}} \geq 2 \sum_{i=1}^{k}\sum_{j=k+1}^{N} J(|i-j|).	
\end{equation}
Under this restriction, the equality in the equation above holds if and only if $\sigma = L^{(k)}$ or $\sigma = {R}^{(k)}$.
\end{prop}

\begin{proof}
Let us prove the result by induction. Let $\mathcal{H}_{N}$ be defined by
\begin{equation}
\mathcal{H}_{N}(\sigma_{1}, \dots,\sigma_{N}) = \sum_{i=1}^{N}\sum_{j=1}^{N} J(|i-j|) \mathds{1}_{\{\sigma_{i} \neq \sigma_{j}\}} = 2 \sum_{i :\,\sigma_{i} = 1}\sum_{j :\,\sigma_{j} = -1} J(|i-j|).	
\end{equation}
Note that the result is trivial if $N=1$. Assuming that it holds for $N \geq 1$, let us prove that it also holds for $N+1$. 
In case $\sigma_{1} = 1$, applying our induction hypothesis and Lemma \ref{lero}, we have
\begin{eqnarray}
\mathcal{H}_{N+1}(1,\sigma_{2}, \dots,\sigma_{N+1}) &=& 
2 \sum_{j = 1}^{N} J(j)  \mathds{1}_{\{\sigma_{j+1} = -1\}} +  \mathcal{H}_{N}(\sigma_{2}, \dots,\sigma_{N+1}) \\
\label{awp1}&\geq& 2 \sum_{j = k}^{N} J(j) + 2 \sum_{i=1}^{k-1}\sum_{j=k}^{N} J(|i-j|)\\
&=& 2 \sum_{i=1}^{k}\sum_{j=k+1}^{N+1} J(|i-j|).
\end{eqnarray}
Replacing the inequality sign in equation (\ref{awp1}) by an equality, it follows that
\begin{equation}
0 \leq \mathcal{H}_{N}(\sigma_{2}, \dots,\sigma_{N+1}) - 2 \sum_{i=1}^{k-1}\sum_{j=k}^{N} J(|i-j|)
= 2 \sum_{j = k}^{N} J(j) - 2 \sum_{j = 1}^{N} J(j)  \mathds{1}_{\{\sigma_{j+1} = -1\}} \leq 0,
\end{equation}
hence,
\begin{equation}
\sum_{j = 1}^{k-1} J(j) - \sum_{j = 1}^{N} J(j)  \mathds{1}_{\{\sigma_{j+1} = 1\}} = 0.
\end{equation}
Using Lemma \ref{lero} again, we conclude that $\sigma_{j} = 1$ whenever $1 \leq j \leq k$,
and $\sigma_{j} = -1$ whenever $k+1 \leq j \leq N+1$.
Now, in case $\sigma_{1}  = -1$, we write $\mathcal{H}_{N+1}(-1,\sigma_{2}, \dots,\sigma_{N+1})$ as 
\begin{eqnarray}
\mathcal{H}_{N+1}(-1,\sigma_{2}, \dots,\sigma_{N+1}) = \mathcal{H}_{N+1}(1,-\sigma_{2}, \dots,-\sigma_{N+1})
\end{eqnarray}
and apply our previous result in order to obtain 
\begin{equation}
\mathcal{H}_{N+1}(-1,\sigma_{2}, \dots,\sigma_{N+1}) \geq 2 \sum_{i=1}^{N+1-k}\sum_{j=N+2-k}^{N+1} J(|i-j|) = 2 \sum_{i=1}^{k}\sum_{j=k+1}^{N+1} J(|i-j|),	
\end{equation}
where the equality holds only if $\sigma_{j} = -1$ whenever $1 \leq j \leq N+1-k$,
and $\sigma_{j} = 1$ whenever $N+2-k \leq j \leq N+1$.
\end{proof}

As an immediate consequence of Proposition \ref{csgo} the next results follows.

\begin{teo}
\label{csgo2}
Given an integer $k \in \{0, \dots, N\}$, if we restrict to all $\sigma\in\mathcal{M}_k$, then	
\begin{equation}
H_{\Lambda, h}(\sigma) \geq 2 \sum_{i=1}^{k}\sum_{j=k+1}^{N} J(|i-j|) + h (N - 2k) - \frac{1}{2}\sum_{i=1}^{N}\sum_{j =1}^{N} J(|i-j|).	
\end{equation}
Under this restriction, the equality in the equation above holds if and only if $\sigma = R^{(k)}$ or $\sigma = {L}^{(k)}$
\end{teo}

\subsection{Proof of  Theorem \ref{t:minmax}.\ref{minmax1}(minimax)}


\begin{proof}[Proof of Theorem \ref{t:minmax}.\ref{minmax1}]
Define $f: \{0,\dots, N\} \rightarrow \mathbb{R}$ as 
\begin{equation}
f(k) = H_{\Lambda, h} (\mathscr{P}{^{(k)}}).	
\end{equation}	
It follows that 
\begin{eqnarray*}
\Delta f (k) &=& f(k+1) - f(k) \\
&=& 2\left( \sum_{i=1}^{k+1} \sum_{j=k+2}^{N} J(|i-j|) - \sum_{i=1}^{k} \sum_{j=k+1}^{N} J(|i-j|) - h \right) \\
&=& 2\left( \sum_{j=k+2}^{N} J(|k+1-j|) + \sum_{i=1}^{k} \sum_{j=k+2}^{N} J(|i-j|) - \sum_{i=1}^{k} \sum_{j=k+1}^{N} J(|i-j|) - h \right) \\	
&=& 2\left( \sum_{j=k+2}^{N} J(|k+1-j|) - \sum_{i=1}^{k} J(|i-(k+1)|) - h \right) \\
&=&  2\left( \sum_{i=1}^{N-k-1} J(i) - \sum_{i=1}^{k} J(i) - h \right) 
\end{eqnarray*}
holds for all $k$ such that $0 \leq k \leq N-1$, and 
\begin{eqnarray*}
\Delta^{2} f (k) &=& \Delta f(k+1) - \Delta f(k) \\
&=& 2\left( \sum_{i=1}^{N-k-2} J(i) - \sum_{i=1}^{N-k-1} J(i) - \sum_{i=1}^{k+1} J(i) + \sum_{i=1}^{k} J(i) \right) \\
&=& -2(J(N-k-1) + J(k+1))	
\end{eqnarray*}
holds whenever  $0 \leq k \leq N-2$. 

Note that 
\begin{equation}\label{omg1}
\Delta f (0) = 2\left( \sum_{i=1}^{N-1} J(i) - h \right) > 0,   
\end{equation}
$1 \leq \left\lfloor \frac{N}{2} \right\rfloor\leq N-1$, and
\begin{equation}\label{omg2}
\Delta f \left(\left\lfloor \frac{N}{2} \right\rfloor\right) < 0. 
\end{equation}
It follows from $\Delta^{2} f < 0 $ and equations (\ref{omg1}) and (\ref{omg2}) that $f$ satisfies 
\begin{equation}\label{gammapos}
f(0) < f(1)
\end{equation}
and
\begin{equation}
f\left(\left\lfloor \frac{N}{2} \right\rfloor\right) > \dots > f(N),	
\end{equation}
therefore, $f(k_{0}) = \max_{0 \leq k \leq N} f(k)$ for some $k_{0} \in \{1,\dots, \left\lfloor \frac{N}{2} \right\rfloor\}$.

Defining the path $\gamma: \mathbf{-1} \rightarrow \mathbf{+1}$ by $\gamma = (L^{(0)}, L^{(1)}, \dots, L^{(N)})$, it is easy to see that
\begin{equation}
\Phi(\mathbf{-1}, \mathbf{+1}) = \max_{\sigma \in \gamma}H_{\Lambda,h}(\sigma) =  \max_{0 \leq k \leq N}H_{\Lambda, h} (\mathscr{P}{^{(k)}}) =\Gamma+H_{\Lambda,h}(\mathbf{-1}). 	
\end{equation}
\end{proof}

\subsection{Proof of Theorem \ref{t:minmax}.\ref{minmax3} and  \ref{t:minmax}.\ref{minmax2}}
Before giving the proof of the second point of the main theorem, we give some results about the control of the energy of a spin-flipped configuration.
Given a configuration $\sigma$ and $k \in \Lambda$,  the \emph{spin-flipped}  configuration $\theta_{k}\sigma$ is defined as:
\begin{equation}
(\theta_{k}\sigma)_{i} = 
\begin{cases}
-\sigma_{k} &\text{if $i = k$, and}\\
\sigma_{i} &\text{otherwise.}	
\end{cases}	
\end{equation}
Note that the energetic cost to flip the spin at position $k$ from the configuration $\sigma$ is given by
\begin{eqnarray*}
H_{\Lambda,h}(\theta_{k}\sigma) - H_{\Lambda,h}(\sigma) &=& \sum_{\{i,j\} \subseteq \Lambda} J(|i-j|)(\sigma_{i}\sigma_{j} - (\theta_{k}\sigma)_{i}(\theta_{k}\sigma)_{j}) 
+ h \sum_{i \in \Lambda} (\sigma_{i} - (\theta_{k}\sigma)_{i}) \\
&=& \left(\sum_{j \in \Lambda} J(|k-j|)2\sigma_{k}\sigma_{j}  
+ 2h \sigma_{k} \right) \\
&=& 2\sigma_{k}\left(\sum_{j \in \Lambda} J(|k-j|)\sigma_{j} + h \right).
\end{eqnarray*}

\begin{prop}
Under Condition~\ref{c:condition}, given a configuration $\sigma$ such that
\begin{equation}
\label{property} 
H_{\Lambda,h}(\theta_{k}\sigma) - H_{\Lambda,h}(\sigma) \geq 0 
\end{equation}
for every $k \in \{1,\dots,N\}$, then 
either $\sigma = \mathbf{-1}$ or $\sigma = \mathbf{+1}$.	
\end{prop}

\begin{proof}
Let $k \in \{1,\dots,N-1\}$, and let $\sigma$ be a configuration such that $\sigma_{i}= +1$ whenever $1 \leq i \leq k$ and $\sigma_{k+1} = -1$. In the following,
we show that every such $\sigma$ cannot satisfy property (\ref{property}). If property (\ref{property}) is satisfied, then
\begin{equation}
\begin{cases}
H_{\Lambda,h}(\theta_{k}\sigma) - H_{\Lambda,h}(\sigma) \geq 0 \\
H_{\Lambda,h}(\theta_{k+1}\sigma) - H_{\Lambda,h}(\sigma) \geq 0	
\end{cases}	
\end{equation}
that is,
\begin{equation}
\begin{cases}
\sum_{i=1}^{k-1}J(|k-i|) - J(1) + \sum_{i=k+2}^{N}J(|k-i|)\sigma_{i} + h \geq 0 \\
-\left(\sum_{i=1}^{k}J(|k+1-i|) + \sum_{i=k+2}^{N}J(|k+1-i|)\sigma_{i} + h \right)\geq 0.	
\end{cases}
\end{equation}	
Summing both equations above, we have
\begin{eqnarray*}
0 &\leq& -J(k) - J(1) + \sum_{i=k+2}^{N}(J(i-k) - J(i-k-1))\sigma_{i} \\
&\leq&	-J(k) - J(1) + \sum_{i=k+2}^{N}(J(i-k-1) - J(i-k))\\
&=& -J(k) - J(1) + \sum_{i=1}^{N-k-1}(J(i) - J(i+1))\\
&=& -J(k)-J(N-k)
\end{eqnarray*}	
that is a contradiction. Analogously, every configuration $\sigma$ such that such that $\sigma_{i}= -1$ whenever $1 \leq i \leq k$ and $\sigma_{k+1} = 1$ for some
$k \in \{1,\dots,N-1\}$, property (\ref{property}) cannot be satisfied. Therefore, we conclude that for every $\sigma$ different from 
$\mathbf{-1}$ and $\mathbf{+1}$, property (\ref{property}) does not hold.

The proof of the converse statement is straightforward.
\end{proof}

As an immediate consequence of the result above, the next result follows.
\begin{cor}\label{path}
Under Condition~\ref{c:condition}, for every configuration $\sigma$ different from
$\mathbf{-1}$ and $\mathbf{+1}$, there is a path $\gamma = (\sigma^{(1)}, \dots, \sigma^{(n)})$, where $\sigma^{(1)} = \sigma$ and $\sigma^{(n)} \in \{\mathbf{-1},\mathbf{+1}\}$,
such that $H_{\Lambda,h}(\sigma^{(i+1)}) < H_{\Lambda,h}(\sigma^{(i)})$.
\end{cor}

We have now all the element for proving item \ref{minmax3} and \ref{minmax2} of Theorem ~\ref{t:minmax}.

\begin{proof}[Proof of Theorem \ref{t:minmax}.\ref{minmax3}]
First, note that it follows from inequality (\ref{gammapos}) that $\Gamma > 0$. Now, let us show that $V_{\mathbf{-1}}$ satisfies
\begin{equation}
V_{\mathbf{-1}} = \Phi(\mathbf{-1},\mathbf{+1}) - H_{\Lambda,h}(\mathbf{-1}).	
\end{equation} 
Since $\mathbf{+1}\in \mathscr{I}_{\mathbf{-1}}$, we have 
\begin{equation}
\label{v1}
V_{\mathbf{-1}} \leq \Phi(\mathbf{-1},\mathbf{+1}) - H_{\Lambda,h}(\mathbf{-1}).
\end{equation}
So, we conclude the proof if we show that 
\begin{equation}
\label{v2}
\Phi(\mathbf{-1},\mathbf{+1}) \leq \Phi(\mathbf{-1},\eta)
\end{equation}
 holds for every $\eta\in \mathscr{I}_{\mathbf{-1}}$. Let $\gamma_{1}: \mathbf{-1} \rightarrow \eta$ be a path from $\mathbf{-1}$ to $\eta$ given by
$\gamma_{1} = (\sigma^{(1)}, \dots, \sigma^{(n)})$,
then, according to Corollary \ref{path}, there is a path
$\gamma_{2} : \eta \rightarrow \mathbf{+1}$, say $\gamma_{2} = (\eta^{(1)},\dots, \eta^{(m)})$, along which the energy decreases. Hence, the path 
$\gamma: \mathbf{-1} \rightarrow \mathbf{+1}$ given by 
\begin{equation}
\gamma = (\sigma^{(1)},\dots, \sigma^{(n-1)}, \eta^{(1)},\dots, \eta^{(m)})	
\end{equation}	
satisfies
\begin{equation}
\Phi_{\gamma} (-\mathbf{1},+\mathbf{1}) = \Phi_{\gamma_1} (-\mathbf{1},\eta) 
\vee \Phi_{\gamma_2} (\eta,+\mathbf{1}))
=\Phi_{\gamma_1} (-\mathbf{1},\eta).    	
\end{equation}
Hence, the inequality 
\begin{equation}
\Phi(\mathbf{-1},\mathbf{+1}) \leq \Phi_{\gamma_{1}}(\mathbf{-1},\eta)
\end{equation}
holds for every path $\gamma_{1}: \mathbf{-1} \rightarrow \eta$, and equation (\ref{v2}) follows.
\end{proof}

\begin{proof}[Proof of Theorem \ref{t:minmax}.\ref{minmax2}]
Given $\sigma \notin \{\mathbf{-1}, \mathbf{+1}\}$, let us show now that
\begin{equation}
\Phi(\sigma,\eta) - H_{\Lambda,h}(\sigma) < V_{\mathbf{-1}}	
\end{equation}
holds for any $\eta\in \mathscr{I}_{\sigma}$. Let us consider the following cases.
\begin{enumerate}
\item Case $\eta = \mathbf{+1}$. According to Corollary (\ref{path}), there is a path $\gamma = (\sigma^{(1)}, \dots, \sigma^{(n)})$ from 
$\sigma^{(1)}= \sigma$ to $\sigma^{(n)} \in \{\mathbf{-1},\mathbf{+1}\}$ along which the energy decreases. 
\begin{itemize}
\item[(a)] If $\sigma^{(n)} = \mathbf{-1}$, then the path $\gamma_{0}: \sigma \rightarrow \mathbb{\eta}$	
given by $\gamma_{0} = (\sigma^{(1)}, \dots, \sigma^{(n-1)}, L^{(0)}, \dots, L^{(N)})$ satisfies
\begin{eqnarray*}
\Phi(\sigma,\eta) - H_{\Lambda,h}(\sigma) &\leq& \max_{\zeta \in \gamma_{0}}H_{\Lambda,h}(\zeta) - H_{\Lambda,h}(\sigma) \\
&\leq& \left(\max_{\zeta \in \gamma}H_{\Lambda,h}(\zeta)\right) 
\vee \left(\max_{0 \leq k \leq N}H_{\Lambda,h}(L^{(k)})\right) - H_{\Lambda,h}(\sigma)	\\
&=&  0
\vee \left(\max_{0 \leq k \leq N}H_{\Lambda,h}(L^{(k)}) - H_{\Lambda,h}(\sigma)\right) \\
&<& \max_{0 \leq k \leq N}H_{\Lambda,h}(L^{(k)}) - H_{\Lambda,h}(\mathbf{-1}) \\
&=& V_{\mathbf{-1}}.
\end{eqnarray*}

\item[(b)] Otherwise, if $\sigma^{(n)} = \mathbf{+1}$, then
\begin{eqnarray*}
\Phi(\sigma,\eta) - H_{\Lambda,h}(\sigma) &\leq& \max_{\zeta \in \gamma}H_{\Lambda,h}(\zeta) - H_{\Lambda,h}(\sigma) \\
&=& 0 \\
&<& V_{\mathbf{-1}}.
\end{eqnarray*}
\end{itemize}

\item Case $\eta = \mathbf{-1}$. According to Corollary (\ref{path}), there is a path $\gamma = (\sigma^{(1)}, \dots, \sigma^{(n)})$ from 
$\sigma^{(1)}= \sigma$ to $\sigma^{(n)} \in \{\mathbf{-1},\mathbf{+1}\}$ along which the energy decreases. 
\begin{itemize}
\item[(a)] If $\sigma^{(n)} = \mathbf{+1}$, then the path $\gamma_{0}: \sigma \rightarrow \mathbb{\eta}$	
given by $\gamma_{0} = (\sigma^{(1)}, \dots, \sigma^{(n-1)}, L^{(N)}, \dots, L^{(0)})$ satisfies
\begin{eqnarray*}
\Phi(\sigma,\eta) - H_{\Lambda,h}(\sigma) &\leq& \max_{\zeta \in \gamma_{0}}H_{\Lambda,h}(\zeta) - H_{\Lambda,h}(\sigma) \\
&\leq& \left(\max_{\zeta \in \gamma}H_{\Lambda,h}(\zeta)\right) 
\vee \left(\max_{0 \leq k \leq N}H_{\Lambda,h}(L^{(k)})\right) - H_{\Lambda,h}(\sigma)	\\
&=&  0
\vee \left(\max_{0 \leq k \leq N}H_{\Lambda,h}(L^{(k)}) - H_{\Lambda,h}(\sigma)\right) \\
&<& \max_{0 \leq k \leq N}H_{\Lambda,h}(L^{(k)}) - H_{\Lambda,h}(\mathbf{-1}) \\
&=& V_{\mathbf{-1}}.
\end{eqnarray*}

\item[(b)] Otherwise, if $\sigma^{(n)} = \mathbf{-1}$, then
\begin{eqnarray*}
\Phi(\sigma,\eta) - H_{\Lambda,h}(\sigma) &\leq& \max_{\zeta \in \gamma}H_{\Lambda,h}(\zeta) - H_{\Lambda,h}(\sigma) \\
&=& 0 \\
&<& V_{\mathbf{-1}}.	
\end{eqnarray*}
\end{itemize}

\item Case $\eta \notin \{\mathbf{-1}, \mathbf{+1}\}$. Let $\gamma_{1} = (\sigma^{(1)}, \dots, \sigma^{(n)})$ and $\gamma_{2} = (\eta^{(1)}, \dots, \eta^{(m)})$ be
paths from $\sigma^{(1)} = \sigma$ to $\sigma^{(n)} \in \{\mathbf{-1},\mathbf{+1}\}$ and from $\eta^{(1)} = \eta$ to $\eta^{(m)} \in \{\mathbf{-1},\mathbf{+1}\}$, respectively, along which
the energy decreases. 
\begin{itemize}
\item[(a)] If $\sigma^{(n)} = \eta^{(m)}$, define the path $\gamma: \sigma \rightarrow \eta$ given by $\gamma_{0} = (\sigma^{(1)}, \dots, \sigma^{(n-1)}, \eta^{(m)},\dots, \eta^{(1)})$
in order to obtain
\begin{eqnarray*}
\Phi(\sigma,\eta) - H_{\Lambda,h}(\sigma) &\leq& \max_{\zeta \in \gamma_{0}}H_{\Lambda,h}(\zeta) - H_{\Lambda,h}(\sigma) \\
&=& \left(\max_{\zeta \in \gamma_{1}}H_{\Lambda,h}(\zeta)\right) 
\vee \left(\max_{\zeta \in \gamma_{2}}H_{\Lambda,h}(\zeta)\right) - H_{\Lambda,h}(\sigma)	\\
&=&  H_{\Lambda,h}(\sigma)
\vee H_{\Lambda,h}(\eta) - H_{\Lambda,h}(\sigma) \\
&=& 0 \\
&<& V_{\mathbf{-1}}.
\end{eqnarray*}
\item[(b)] If $\sigma^{(n)} = \mathbf{-1}$ and $\eta^{(m)} = \mathbf{+1}$, let us define the path $\gamma_{0}: \sigma \to \eta$ given by 
\begin{equation}
\gamma_{0} = (\sigma^{(1)}, \dots, \sigma^{(n-1)}, L^{(0)}, \dots, L^{(N)}, \eta^{(m-1)}, \dots, \eta^{(1)})	
\end{equation}
satisfies
\begin{eqnarray*}
\Phi(\sigma,\eta) - H_{\Lambda,h}(\sigma) &\leq& \max_{\zeta \in \gamma_{0}}H_{\Lambda,h}(\zeta) - H_{\Lambda,h}(\sigma) \\
&=& \left(\max_{\zeta \in \gamma_{1}}H_{\Lambda,h}(\zeta)\right) 
\vee \left(\max_{0 \leq k \leq N}H_{\Lambda,h}(L^{(k)})\right) \vee  \left(\max_{\zeta \in \gamma_{2}}H_{\Lambda,h}(\zeta)\right)  - H_{\Lambda,h}(\sigma)	\\
&=&  H_{\Lambda,h}(\sigma)
\vee \left(\max_{0 \leq k \leq N}H_{\Lambda,h}(L^{(k)})\right) \vee  H_{\Lambda,h}(\eta)  - H_{\Lambda,h}(\sigma)	\\
&=&  0
\vee \left(\max_{0 \leq k \leq N}H_{\Lambda,h}(L^{(k)}) - H_{\Lambda,h}(\sigma)\right) \\
&<& \max_{0 \leq k \leq N}H_{\Lambda,h}(L^{(k)}) - H_{\Lambda,h}(\mathbf{-1}) \\
&=& V_{\mathbf{-1}}.
\end{eqnarray*}

\item[(c)] If $\sigma^{(n)} = \mathbf{+1}$ and $\eta^{(m)} = \mathbf{-1}$, let us define the path $\gamma_{0}: \sigma \to \eta$ given by 
\begin{equation}
\gamma_{0} = (\sigma^{(1)}, \dots, \sigma^{(n-1)}, L^{(N)}, \dots, L^{(0)}, \eta^{(m-1)}, \dots, \eta^{(1)})	
\end{equation}
satisfies
\begin{eqnarray*}
\Phi(\sigma,\eta) - H_{\Lambda,h}(\sigma) &\leq& \max_{\zeta \in \gamma_{0}}H_{\Lambda,h}(\zeta) - H_{\Lambda,h}(\sigma) \\
&=& \left(\max_{\zeta \in \gamma_{1}}H_{\Lambda,h}(\zeta)\right) 
\vee \left(\max_{0 \leq k \leq N}H_{\Lambda,h}(L^{(k)})\right) \vee  \left(\max_{\zeta \in \gamma_{2}}H_{\Lambda,h}(\zeta)\right)  - H_{\Lambda,h}(\sigma)	\\
&=&  H_{\Lambda,h}(\sigma)
\vee \left(\max_{0 \leq k \leq N}H_{\Lambda,h}(L^{(k)})\right) \vee  H_{\Lambda,h}(\eta)  - H_{\Lambda,h}(\sigma)	\\
&=&  0
\vee \left(\max_{0 \leq k \leq N}H_{\Lambda,h}(L^{(k)}) - H_{\Lambda,h}(\sigma)\right) \\
&<& \max_{0 \leq k \leq N}H_{\Lambda,h}(L^{(k)}) - H_{\Lambda,h}(\mathbf{-1}) \\
&=& V_{\mathbf{-1}}.
\end{eqnarray*}
\end{itemize}
\end{enumerate}

We conclude that for every $\sigma \notin \{\mathbf{-1},\mathbf{+1}\}$, we have $V_{\sigma} < V_{\mathbf{-1}}$.
\end{proof}

\section{Proofs of the critical droplets results}
\begin{proof}[Proof of Proposition~\ref{critdrop}]
 As in the proof of  Theorem~\ref{t:minmax}, let us define $f: \{0,\dots, N\} \rightarrow \mathbb{R}$ as 
\begin{equation}
f(i) = H_{\Lambda, h} (L^{(i)}),	
\end{equation}
and recall that 
\begin{equation}
  \Delta f (i) = 2\left( \sum_{n=1}^{N-i-1} J(n) - \sum_{n=1}^{i} J(n) - h \right). 
\end{equation}
In the first case, we have $\Delta f(L-1) = 2(h_{L-1}^{(N)} - h) > 0$,
thus, since $f$ decreases for all $i$ greater than $L$,  and since $\Delta^2 f<0$, we conclude that $f$ attains a unique strict global maximum at $L$. In the second case,
we have $\Delta f(k-1) = 2(h_{k-1}^{(N)} - h) > 0$ and $\Delta f(k) = 2(h_{k}^{(N)} - h) < 0$,
so, $f$ attains a unique strict global maximum at $k$. Finally, in the third case, we have $\Delta f(k) = 0$,
that is, $f(k) = f(k+1)$. Using the fact that $\Delta f(k+1) < 0 < \Delta f(k-1)$, we conclude that the global maximum of $f$ can we only be reached 
at $k$ and $k+1$.
\end{proof}

\begin{proof}[Proof of Corollary~\ref{corcrit}]
Since $\sum_{n = 1}^{\infty}J(n)$ converges, it follows that the set in equation (\ref{crit}) is nonempty, thus $k_{c}$ is well defined. Then, we have
\begin{equation}
\sum_{n = k_{c}+1}^{\infty}J(n) \leq h < \sum_{n = k_{c}}^{\infty}J(n).
\end{equation}
For all $N$ sufficiently large such that $\left\lfloor \frac{N}{2} \right\rfloor > k_{c}$ and
\begin{equation}
 \sum_{n = N - k_{c}+1}^{\infty}J(n) < \sum_{n = k_{c}}^{\infty}J(n) - h,
\end{equation}
we have
\begin{equation}
 h < \sum_{n = k_{c}}^{\infty}J(n) - \sum_{n = N - k_{c}+1}^{\infty}J(n) = h_{k_{c}-1}^{(N)}
\end{equation}
and
\begin{equation}
 h_{k_{c}}^{(N)} = \sum_{n = k_{c}+1}^{\infty}J(n) - \sum_{n = N - k_{c}}^{\infty}J(n) < h.
\end{equation}
Therefore, by means of Proposition \ref{critdrop}, we conclude that for $N$ large enough, $k_{c}$ satisfies
\begin{equation}
 H_{\Lambda,h}(\mathscr{P}^{(k_{c})}) > \max_{\substack{0 \leq i \leq N \\ i \neq k_{c}}} H_{\Lambda,h}(\mathscr{P}^{(i)}).
\end{equation} 
\end{proof}


\appendix

\section{Appendix}

\begin{lemma}\label{lero}
Let $\Lambda$ be a finite subset of $\mathbb{N}$, then
\begin{equation}
\sum_{i \in \Lambda} J(i) \leq \sum_{i = 1}^{\# \Lambda} J(i),	
\end{equation} 	
moreover, the equality holds if and only if $\Lambda = \{1, \dots, \# \Lambda\}$.
\end{lemma}

\begin{proof}
Let $k$ be the number of elements of $\Lambda$. Note that for $k=0$ the result holds, so, suppose that it holds whenever $\Lambda$ has $k$ elements. 
Given a subset $\Lambda$ of $\mathbb{N}$ containing $k+1$ elements, let $k_{0}$ be its the maximal element, then, using our induction
hypothesis and the fact that $k_{0} \geq k + 1$,
we have
\begin{equation}\label{chapulin}
\sum_{i \in \Lambda} J(i) =  J(k_{0}) + \sum_{i \in \Lambda \backslash \{k_{0}\}} J(i) \leq J(k+1)  
+ \sum_{i = 1}^{k} J(i) = \sum_{i = 1}^{k+1} J(i).	
\end{equation}
In case we have an equality in equation (\ref{chapulin}), we have
\begin{equation}
0 \leq \sum_{i = 1}^{k} J(i) - \sum_{i \in \Lambda \backslash \{k_{0}\}} J(i) = J(k_{0}) - J(k+1) \leq 0,	
\end{equation}
thus, $\Lambda \backslash \{k_{0}\} = \{1,\dots,k\}$ and $k_{0} = k+1$.
\end{proof}


\bibliographystyle{abbrv}

\end{document}